%% file: ms.tex
\documentclass[a4paper,USenglish]{lipics-v2018}
 
\usepackage{microtype}

\usepackage{enumerate}
\usepackage{comment}
\usepackage{amsmath,amssymb}

\graphicspath{{./fig/}}

\nolinenumbers
\hideLIPIcs

\title{Incremental Optimization of Independent Sets under Reachability Constraints}

\author{Takehiro Ito}{Tohoku University, Japan}{takehiro@ecei.tohoku.ac.jp}{}{}
\author{Haruka Mizuta}{Tohoku University, Japan}{haruka.mizuta.s4@dc.tohoku.ac.jp}{}{}
\author{Naomi Nishimura}{University of Waterloo, Canada}{nishi@uwaterloo.ca}{}{}
\author{Akira Suzuki}{Tohoku University, Japan}{a.suzuki@ecei.tohoku.ac.jp}{}{}

\authorrunning{T. Ito, H. Mizuta, N. Nishimura, and A. Suzuki} 

\Copyright{Takehiro Ito, Haruka Mizuta, Naomi Nishimura, and Akira Suzuki}

\subjclass{}

\EventEditors{John Q. Open and Joan R. Acces}
\EventNoEds{2}
\EventLongTitle{42nd Conference on Very Important Topics (CVIT 2016)}
\EventShortTitle{CVIT 2016}
\EventAcronym{CVIT}
\EventYear{2016}
\EventDate{December 24--27, 2016}
\EventLocation{Little Whinging, United Kingdom}
\EventLogo{}
\SeriesVolume{42}
\ArticleNo{}

\input{newcom}

\begin{document}

\maketitle

\begin{abstract}
	We introduce a new framework for reconfiguration problems, and apply it to independent sets as the first example. 
	Suppose that we are given an independent set $\ind_\ini$ of a graph $G$, and an integer $\thr \ge 0$ which represents a lower bound on the size of any independent set of $G$.
	Then, we are asked to find an independent set of $G$ having the maximum size among independent sets that are reachable from $\ind_\ini$ by either adding or removing a single vertex at a time such that all intermediate independent sets are of size at least $\thr$.
	We show that this problem is PSPACE-hard even for bounded pathwidth graphs, and remains NP-hard for planar graphs.
	On the other hand, we give a linear-time algorithm to solve the problem for chordal graphs.
	We also study the fixed-parameter (in)tractability of the problem with respect to the following three parameters: the degeneracy $\dg$ of an input graph, a lower bound $\thr$ on the size of the independent sets, and a lower bound $\solsize$ on the solution size. 
	We show that the problem is fixed-parameter intractable when only one of $\dg$, $\thr$, and $\solsize$ is taken as a parameter. 
	On the other hand, we give a fixed-parameter algorithm when parameterized by $\solsize+\dg$; this result implies that the problem parameterized only by $\solsize$ is fixed-parameter tractable for planar graphs, and for bounded treewidth graphs. 
\end{abstract}

\section{Introduction}\label{sec:intro}
Recently, the reconfiguration framework~\cite{IDHPSUU11} has been intensively applied to a variety of search problems.
(See, e.g., surveys~\cite{van13,Nis17}.) 
For example, the {\sc independent set reconfiguration} problem is one of the most well-studied reconfiguration problems~\cite{BB17,Bon14,BMP17,HU16,IKO14,IKOSUY14,KMM12,LM18,LMPRS15,MNRSS17,W18,vdZ15}. 
For a graph $G$, a vertex subset $\ind \subseteq V(G)$ is an {\em independent set} of $G$ if no two vertices in $\ind$ are adjacent in $G$.
Suppose that we are given two independent sets $\ind_\ini$ and $\ind_\tar$ of $G$, and imagine that a token (coin) is placed on each vertex in $\ind_\ini$.
Then, for an integer lower bound $\thr \ge 0$, {\sc independent set reconfiguration} {\em under the $\TARrule$ rule} is the problem of determining whether we can transform $\ind_\ini$ into $\ind_\tar$ via independent sets of size at least $\thr$ such that each intermediate independent set can be obtained from the previous one by either adding or removing a single token.\footnote{$\TARrule$ stands for Token Addition and Removal, and there are two other well-studied reconfiguration rules called $\TS$ (Token Sliding) and $\TJ$ (Token Jumping)~\cite{KMM12}. We omit the details in this paper.}	
In the example of \figurename~\ref{fig:example}, $\ind_\ini$ can be transformed into $\ind_\tar = \ind_3$ via the sequence $\seq{\ind_\ini, \ind_1, \ind_2, \ind_3}$, but not into $\ind_\tar^\prime$, when $\thr = 1$.

\begin{figure}[t]
	\centering
	\includegraphics[width=0.9\linewidth]{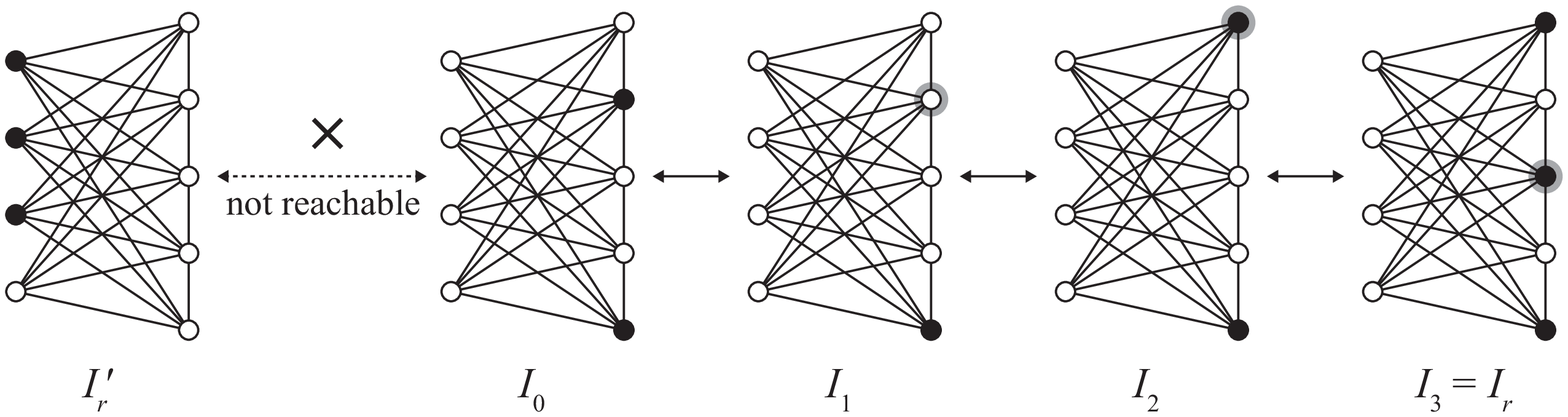}
	\caption{A sequence $\seq{\ind_\ini, \ind_1, \ind_2, \ind_3}$ of independent sets under the $\TARrule$ rule for the lower bound $\thr=1$, where the vertices in independent sets are colored with black.}
	\label{fig:example}
\end{figure}

Like this problem, many reconfiguration problems have the following basic structure: 
we are given two feasible solutions of an original search problem, and are asked to determine whether we can transform one into the other by repeatedly applying a specified reconfiguration rule while maintaining feasibility. 
These kinds of reconfiguration problems model several ``dynamic'' situations of systems, where we wish to find a step-by-step transformation from the current configuration of a system into a more desirable one. 

However, it is not easy to obtain a more desirable configuration for an input of a reconfiguration problem, because many original search problems are NP-hard. 
Furthermore, there may exist (possibly, exponentially many) desirable configurations; even if we can not reach a given target from the current configuration, there may exist another desirable configuration which is reachable. 
Recall the example of \figurename~\ref{fig:example}, where both $\ind_\tar$ and $\ind_\tar^\prime$ have the same size three (which is larger than that of the current independent set $\ind_\ini$), but $\ind_\ini$ can reach only $\ind_\tar$.

\subsection{Our problem}

In this paper, we propose a new framework for reconfiguration problems which asks for a more desirable configuration that is reachable from the current one.
As the first application of this new framework, we consider {\sc independent set reconfiguration} because it is one of the most well-studied reconfiguration problems. 

Suppose that we are given a graph $G$, an integer lower bound $\thr \ge 0$, and an independent set $\ind_\ini$ of $G$. 
Then, we are asked to find an independent set $\ind_\opt$ of $G$ such that $|\ind_\opt|$ is maximized and $\ind_\ini$ can be transformed into $\ind_\opt$ under the $\TARrule$ rule for the lower bound $\thr$.
We call this problem the {\em optimization variant} of {\sc independent set reconfiguration} (denoted by {\sc Opt-ISR}).
To avoid confusion, we call the standard {\sc independent set reconfiguration} problem the {\em reachability variant} (denoted by {\sc Reach-ISR}). 

Note that $\ind_\opt$ is not always a maximum independent set of the graph $G$.
For example, the graph in \figurename~\ref{fig:example} has a unique maximum independent set of size four (consisting of the vertices on the left side), but $\ind_\ini$ cannot be transformed into it.  
Indeed, $\ind_\opt = \ind_3$ for this example when $\thr=1$.

\subsection{Related results}
Although {\sc Opt-ISR} is being introduced in this paper, some previous results 
for {\sc Reach-ISR} are related in the sense that they can be converted into results for {\sc Opt-ISR}. We present such results here, explaining their relation to our results in Sections~\ref{sec:preli}--\ref{sec:FPT}, after formally defining {\sc Opt-ISR} and notation.  

Ito et al.~\cite{IDHPSUU11} showed that {\sc Reach-ISR} under the $\TARrule$ rule is PSPACE-complete.
On the other hand, Kami\'nski et al.~\cite{KMM12} proved that any two independent sets of size at least $\thr+1$ are reachable under the $\TARrule$ rule with the lower bound $\thr$ for even-hole-free graphs.

{\sc Reach-ISR} has been studied well from the viewpoint of fixed-parameter (in)tractability. 
Mouawad et al.~\cite{MNRSS17} showed that {\sc Reach-ISR} under the $\TARrule$ rule is W[1]-hard when parameterized by the lower bound $\thr$ and the length of a desired sequence (i.e., the number of token additions and removals).
Lokshtanov et al.~\cite{LMPRS15} gave a fixed-parameter algorithm to solve {\sc Reach-ISR} under the $\TARrule$ rule when parameterized by the lower bound $\thr$ and the degeneracy $\dg$ of an input graph.

\subsection{Our contributions}

\begin{figure}[t]
	\centering
	\includegraphics[width=0.85\linewidth]{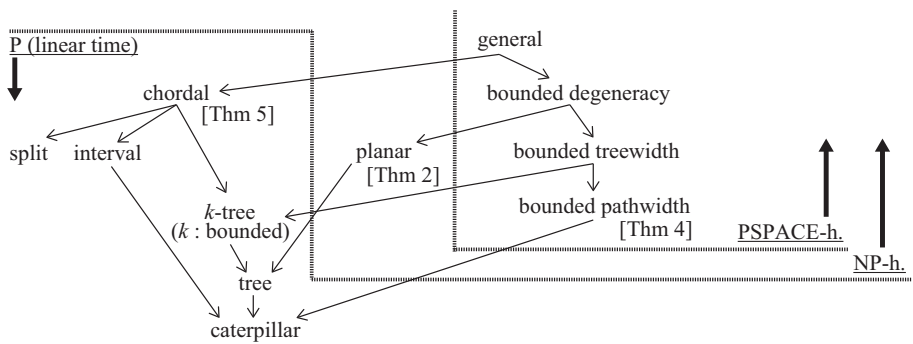}
	\caption{Our results with respect to graph classes.}
	\label{fig:results}
\end{figure}

\begin{table}[t]
	\begin{center}
		\caption{Our results with respect to parameters.}
		\begin{tabular}{|c||c|c|c|}
			\hline
			& (no parameter) & lower bound $\thr$ & solution size $\solsize$ \\
			\hline
			\hline
			(no parameter) & & NP-hard for fixed $\thr$ & W[1]-hard for $\solsize$, {\bf XP} \\
			& --- & (i.e., no FPT, no XP) & (i.e., no FPT) \\
			& & [Corollary~\ref{cor:np-hard}] & [Theorems~\ref{the:w1-hard}, \ref{the:xp}] \\
			\hline
			degeneracy $\dg$ & PSPACE-hard for fixed $\dg$ & NP-hard for fixed $\dg + \thr$ & {\bf FPT} \\
			& (i.e., no FPT, no XP) & (i.e., no FPT, no XP) &  \\
			& [Theorem~\ref{the:pspace_complete}] & [Corollary~\ref{cor:np-hard}] & [Theorem~\ref{the:fpt}] \\
			\hline
		\end{tabular}
		\label{tab:result_parameter}
	\end{center}
\end{table}

In this paper, we study the problem from the viewpoints of polynomial-time solvability and fixed-parameter (in)tractability. 

We first study the polynomial-time solvability of {\sc Opt-ISR} with respect to graph classes, as summarized in \figurename~\ref{fig:results}.
More specifically, we show that {\sc Opt-ISR} is PSPACE-hard even for bounded pathwidth graphs, and remains NP-hard even for planar graphs.
On the other hand, we give a linear-time algorithm to solve the problem for chordal graphs.
We note that our algorithm indeed works in polynomial time for even-hole-free graphs (which form a larger graph class than that of chordal graphs) if the problem of finding a maximum independent set is solvable in polynomial time for even-hole-free graphs; currently, its complexity status is unknown.

We next study the fixed-parameter (in)tractability of {\sc Opt-ISR}, as summarized in \tablename~\ref{tab:result_parameter}.
In this paper, we consider mainly the following three parameters: the degeneracy $\dg$ of an input graph, a lower bound $\thr$ on the size of the independent sets, and the solution size $\solsize$. 
As shown in \tablename~\ref{tab:result_parameter}, we completely analyze the fixed-parameter (in)tractability of the problem according to these three parameters; details are explained below. 

We first consider the problem parameterized by a single parameter. 
We show that the problem is fixed-parameter intractable when only one of $\dg$, $\thr$, and $\solsize$ is taken as a parameter. 
In particular, we prove that {\sc Opt-ISR} is PSPACE-hard for a fixed constant $\dg$ and remains NP-hard for a fixed constant $\thr$, and hence the problem does not admit even an XP algorithm for each single parameter $\dg$ or $\thr$ under the assumption that ${\rm P} \neq {\rm PSPACE}$ or ${\rm P} \neq {\rm NP}$. 
On the other hand, {\sc Opt-ISR} is W[1]-hard for $\solsize$, and admits an XP algorithm with respect to $\solsize$. 

We thus consider the problem taking two parameters. 
However, the problem still remains NP-hard for a fixed constant $\dg+\thr$, and hence it does not admit even an XP algorithm for $\dg+\thr$ under the assumption that ${\rm P} \neq {\rm NP}$. 
Note that the combination of $\thr$ and $\solsize$ is meaningless, since $\thr + \solsize \le 2 \solsize$, as explained in Section~\ref{sec:FPT}. 
On the other hand, we give a fixed-parameter algorithm when parameterized by $\solsize+\dg$; this result implies that {\sc Opt-ISR} parameterized only by $\solsize$ is fixed-parameter tractable for planar graphs, and for bounded treewidth graphs.

\section{Preliminaries} \label{sec:preli}


In this paper, we consider only simple graphs, without loss of generality.
For a graph $G$, we denote by $V(G)$ and $E(G)$ the vertex set and edge set of $G$, respectively.
For a vertex $v \in V(G)$, let $\oneig{G}{v} = \set{ w \in V(G) : vw \in E(G)}$, and let $\cneig{G}{v} = \oneig{G}{v} \cup \set{v}$.
The set $\oneig{G}{v}$ is called the ({\em open}) {\em neighborhood} of $v$ in $G$, while $\cneig{G}{v}$ is called the {\em closed neighborhood} of $v$ in $G$.
For a graph $G$ and a vertex subset $S \subseteq V(G)$, $G[S]$ denotes the subgraph of $G$ {\em induced} by $S$, that is, $V(G[S]) = S$ and $E(G[S]) = \set{ vw \in E(G) : v, w \in S}$.
For a vertex subset $V^\prime \subseteq V(G)$, we simply write $G\setminus V^\prime$ to denote $G[V(G)\setminus V^\prime]$. 
We denote by $\symdiff{A}{B}$ the symmetric difference between two sets $A$ and $B$, that is, $\symdiff{A}{B} = (A \setminus B) \cup (B \setminus A)$.

\subsection*{Optimization variant of {\sc independent set reconfiguration}}

We now formally define our problem. 
For an integer $\thr \ge 0$ and two independent sets $\ind_p$ and $\ind_q$ of a graph $G$ such that $|\ind_p| \geq \thr$ and $|\ind_q| \geq \thr$, a sequence $\mathcal{\ind} = \seq{\ind_1,\ind_2,\ldots,\ind_\ell}$ of independent sets of $G$ is called a {\em reconfiguration sequence} between $\ind_p$ and $\ind_q$ under the $\TARrule$ rule if $\mathcal{\ind}$ satisfies the following three conditions (a)--(c):
\begin{listing}{}
	\item [(a)] $\ind_1 = \ind_p$ and $\ind_\ell = \ind_q$;
	\item [(b)] $\ind_i$ is an independent set of size at least $\thr$ for each $i \in \set{1,2,\ldots,\ell}$; and					
	\item [(c)] $|\symdiff{\ind_i}{\ind_{i+1}}| = 1$ for each $i \in \set{1,2,\ldots,\ell-1}$.
\end{listing}
To emphasize the lower bound $\thr$ on the size of any independent set, we sometimes write $\TAR{\thr}$ instead of $\TARrule$. 
Note that any reconfiguration sequence is {\em reversible}, that is, $\seq{\ind_\ell,\ind_{\ell-1},\ldots,\ind_1}$ is a reconfiguration sequence between $\ind_q$ and $\ind_p$ under the $\TAR{\thr}$ rule. 
We say that two independent sets $\ind_p$ and $\ind_q$ are {\em reachable} under the $\TAR{\thr}$ rule if there exists a reconfiguration sequence between $\ind_p$ and $\ind_q$ under the $\TAR{\thr}$ rule.
We write $\ind_p \sevsteptar{\thr} \ind_q$ if $\ind_p$ and $\ind_q$ are reachable under the $\TAR{\thr}$ rule.

Our problem aims to optimize a given independent set under the $\TARrule$ rule. 
Specifically, the {\em optimization variant} of {\sc independent set reconfiguration} ({\sc Opt-ISR} for short) is defined as follows:
\begin{center}
	\begin{listing}{{\bf Input:}}
		\item[{\bf Input:}] A graph $G$, an integer $\thr \ge 0$, and an independent set $\ind_\ini$ of $G$ such that $|\ind_\ini| \geq \thr$
		\item[{\bf Task:}] Find an independent set $\ind_\opt$ of $G$ such that $\ind_\ini \sevsteptar{\thr} \ind_\opt$ and $|\ind_\opt|$ is maximized. 
	\end{listing}
\end{center}
We denote by a triple $(G,\thr,\ind_\ini)$ an instance of {\sc Opt-ISR}, and call a desired independent set $\ind_\opt$ of $G$ a {\em solution} to $(G,\thr,\ind_\ini)$.
Note that a given independent set $\ind_\ini$ may itself be a solution. 
{\sc Opt-ISR} simply outputs a solution to $(G,\thr,\ind_\ini)$, and does not require the specification of an actual reconfiguration sequence from $\ind_\ini$ to the solution. 

We close this section with noting the following observation which says that {\sc Opt-ISR} for an instance $(G,0,\ind_\ini)$ is equivalent to finding a maximum independent set of $G$. 
\begin{lemma} \label{lem:maxind}
	Every maximum independent set $\indm$ of a graph $G$ is a solution to an instance $(G,0,\ind_\ini)$ of {\sc Opt-ISR}, where $\ind_\ini$ is any independent set of $G$. 
\end{lemma}
\begin{proof}
	Because the lower bound $\thr$ on the size of the independent sets is set to zero, any (maximum) independent set $\ind^\prime$ of $G$ is reachable from any independent set $\ind_\ini$ of $G$, as follows:  
	We first remove all vertices in $\ind_\ini \setminus \ind^\prime$ one by one, and then add all vertices in $\ind^\prime \setminus \ind_\ini$ one by one.
	Thus, the lemma follows.
\end{proof}

\section{Polynomial-Time Solvability} \label{sec:poly}

In this section, we study the polynomial-time solvability of {\sc Opt-ISR}. 

\subsection{NP-hardness for planar graphs}
Lemma~\ref{lem:maxind} implies that results for the {\sc maximum independent set} problem can be applied to {\sc Opt-ISR} for $k=0$. 
For example, we have the following theorem, because {\sc maximum independent set} remains NP-hard even for planar graphs~\cite{GJ79}.
\begin{theorem}\label{the:np-hard-planar}
	{\sc Opt-ISR} is NP-hard for planar graphs and $\thr=0$, where $\thr$ is a lower bound on the size of the independent sets.
\end{theorem}

For an integer $\dg \ge 0$, a graph $G$ is $\dg$-{\em degenerate} if every induced subgraph of $G$ has a vertex of degree at most $\dg$~\cite{LW70}.
The {\em degeneracy} of $G$ is the minimum integer $\dg$ such that $G$ is $\dg$-degenerate.
It is known that the degeneracy of any planar graph is at most five~\cite{LW70}, and hence we have the following corollary.
\begin{corollary}\label{cor:np-hard}
	{\sc Opt-ISR} is NP-hard for $5$-degenerate graphs and $\thr = 0$, where $\thr$ is a lower bound on the size of the independent sets.
\end{corollary}

This corollary implies that {\sc Opt-ISR} admits neither a fixed-parameter algorithm nor an XP algorithm when parameterized by $\dg+\thr$ under the assumption that ${\rm P} \neq {\rm NP}$, where $\dg$ is an upper bound on the degeneracy of an input graph and $\thr$ is a lower bound on the size of the independent sets.
We will discuss the fixed parameter (in)tractability of {\sc Opt-ISR} more deeply in Section~\ref{sec:FPT}.

\subsection{PSPACE-hardness for bounded pathwidth graphs}
In this subsection, we show that {\sc Opt-ISR} is PSPACE-hard even if the pathwidth of an input graph is bounded by a constant.
We first define the pathwidth of a graph, as follows~\cite{RS83}.
A {\em path-decomposition} of a graph $G$ is a sequence $\seq{\bag_1,\bag_2,\ldots,\bag_t}$ of vertex subsets of $V(G)$ such that
\begin{listing}{aaa}
	\item[(a)] for each vertex $u$ of $G$, there exists a subset $\bag_i$ such that $u \in \bag_i$;
	\item[(b)] for each edge $vw$ of $G$, there exists a subset $\bag_j$ such that $v, w \in \bag_j$; and
	\item[(c)] for any three indices  $a, b, c$ such that $a < b < c$, $\bag_a \cap \bag_c \subseteq \bag_b$ holds.
\end{listing}
The {\em pathwidth} of $G$ is the minimum value $p$ such that there exists a path-decomposition $\seq{\bag_1,\bag_2,\ldots,\bag_t}$ of $G$ for which $|\bag_i| \leq p+1$ holds for all $i \in \set{1,2,\ldots,t}$.
A {\em bounded pathwidth graph} is a graph whose pathwidth is bounded by a constant. 
 
The following theorem is the main result of this subsection.
\begin{theorem}\label{the:pspace_complete}
	{\sc Opt-ISR} is PSPACE-hard for bounded pathwidth graphs. 
\end{theorem}
\begin{proof}
	We give a polynomial-time reduction from the (reachability variant of) {\sc maximum independent set reconfiguration} problem ({\sc MISR} for short), defined as follows~\cite{W18}:
	\begin{center}
		\begin{listing}{{\bf Input:}}
			\item[{\bf Input:}] A graph $G^\prime$, and two maximum independent sets $\ind_\ini^\prime$ and $\ind_\tar^\prime$ of $G^\prime$
			\item[{\bf Task:}] Determine whether $\ind_\ini^\prime \sevsteptar{\thr^\prime} \ind_\tar^\prime$ or not, where $\thr^\prime = |\ind_\ini^\prime| -1 = |\ind_\tar^\prime|-1$
		\end{listing}
	\end{center}
	We denote by a triple $(G^\prime,\ind_\ini^\prime,\ind_\tar^\prime)$ an instance of {\sc MISR}.
	This problem is known to be PSPACE-complete for bounded bandwidth graphs~\cite{W18}.
	Since the pathwidth of a graph is at most the bandwidth of the graph~\cite{S99}, {\sc MISR} is PSPACE-complete also for bounded pathwidth graphs.

	\begin{figure}[t]
	\centering
	\includegraphics[width=0.5\linewidth]{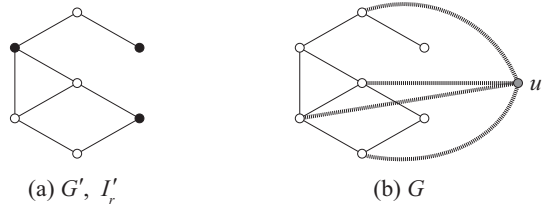}
	\caption{(a) Graph $G^\prime$ and its independent set $\ind_\tar^\prime$ for {\sc MISR}, and (b) the corresponding graph $G$ for {\sc Opt-ISR}, where newly added edges are depicted by thick dotted lines.}
	\label{fig:reduction}
	\end{figure}
	
	Let $(G^\prime, \ind^\prime_\ini, \ind^\prime_\tar)$ be an instance of {\sc MISR} such that the pathwidth of $G^\prime$ is bounded by a constant.
	Then, we construct a corresponding instance $(G, \thr, \ind_\ini)$ of {\sc Opt-ISR}, as follows. (See also \figurename~\ref{fig:reduction}.)
	We add a new vertex $u$ to the graph $G^\prime$, and join it with all vertices in $V(G^\prime) \setminus \ind_\tar^\prime$; let $G$ be the resulting graph, that is, $V(G) = V(G^\prime) \cup \set{u}$ and $E(G) = E(G^\prime) \cup \set{ uv : v \in V(G^\prime) \setminus \ind^\prime_\tar}$.
	Since the pathwidth of $G^\prime$ is bounded by a constant and $V(G) = V(G^\prime) \cup \set{u}$, the pathwidth of $G$ is also bounded by a constant.
	Let $\thr = |\ind^\prime_\ini| -1 = |\ind^\prime_\tar| -1$, and $\ind_\ini = \ind^\prime_\ini$.
	This completes the construction of the corresponding instance $(G, \thr, \ind_\ini)$ of {\sc Opt-ISR}.
	This construction can be done in polynomial time.
\medskip
	
	We now prove the correctness of our reduction.
	We first claim that $G$ has only one maximum independent set, and it is $\ind_\tar^\prime \cup \{u\}$ of size $|\ind_\tar^\prime|+1$. 
	To see this, recall that $\ind_\tar^\prime$ is a maximum independent set of $G^\prime$.
	Therefore, if $G$ has an independent set $\ind$ such that $|\ind| > |\ind_\tar^\prime|$, it must contain $u$.  
	Since $u$ is adjacent to all vertices  in $V(G^\prime) \setminus \ind_\tar^\prime$, only $\ind_\tar^\prime \cup \{u\}$ can be such an independent set of size $|\ind_\tar^\prime|+1$, as claimed. 
	Therefore, to complete the correctness proof of our reduction, we prove that {\sc Opt-ISR} for $(G, \thr, \ind_\ini)$ outputs $\ind_\tar^\prime \cup \{u\}$ if and only if $\ind_\ini^\prime \sevsteptar{\thr} \ind_\tar^\prime$ on $G^\prime$. 
	Since $\ind_\tar^\prime \cup \{u\}$ is only the maximum independent set in $G$, we indeed prove that  $\ind_\ini \sevsteptar{\thr} \ind_\tar^\prime \cup \{u\}$ on $G$ if and only if $\ind_\ini^\prime \sevsteptar{\thr} \ind_\tar^\prime$ on $G^\prime$. 
	
	We first prove the if direction.
	Suppose that $\ind_\ini^\prime \sevsteptar{\thr} \ind_\tar^\prime$ on $G^\prime$.
	Since $G$ contains $G^\prime$ as an induced subgraph, we have $\ind_\ini = \ind^\prime_\ini \sevsteptar{\thr} \ind^\prime_\tar$ on $G$. 
	Then, $\ind_\tar^\prime \cup \{u\}$ can be obtained simply by adding $u$ to $\ind^\prime_\tar$, and hence we can conclude that $\ind_\ini \sevsteptar{\thr} \ind_\tar^\prime \cup \{u\}$ on $G$. 
	
	We next prove the only-if direction.
	Suppose that $\ind_\ini \sevsteptar{\thr} \ind_\tar^\prime \cup \{u\}$ on $G$, that is, there exists a reconfiguration sequence $\mathcal{\ind} = \seq{\ind_\ini,\ind_1,\ldots, \ind_\ell = \ind_\tar^\prime \cup \{u\}}$ on $G$ under the $\TAR{\thr}$ rule.
	Let $\ind_{q+1}$ be the first independent set in $\mathcal{I}$ which contains $u$; notice that $\ind_q = \ind_{q+1} \setminus \set{u}$ because we know $u \not\in \ind_\ini$. 
	Since $\ind_{q+1}$ is an independent set of $G$, no vertex in $\ind_q = \ind_{q+1} \setminus \set{u}$ is adjacent to $u$.
	By the construction of $u$, we thus have $\ind_{q} \subseteq \ind^\prime_\tar$.
	Since $\ind_q$ appears in $\mathcal{\ind}$, we know $|\ind_q| \ge \thr$. 
	Therefore, we can construct a reconfiguration sequence $\mathcal{\ind}_{qr}$ between $\ind_q$ and $\ind_\tar^\prime$ under the $\TAR{\thr}$ rule by simply adding the vertices in $\ind_\tar^\prime \setminus \ind_q$ one by one. 
	Then, by combining $\seq{\ind_\ini,\ind_1,\ldots, \ind_{q-1}}$ and $\mathcal{\ind}_{qr}$ serially, we can obtain a reconfiguration sequence between $\ind_\ini$ and $\ind_\tar^\prime$ under the $\TAR{\thr}$ rule such that all independent sets in the sequence do not contain $u$. 
	We thus have $\ind_\ini^\prime \sevsteptar{\thr} \ind_\tar^\prime$ on $G^\prime$.
\end{proof}

\subsection{Linear-time algorithm for chordal graphs}
A graph $G$ is {\em chordal} if every induced cycle in $G$ is of length three~\cite{BLS99}.
The main result of this subsection is the following theorem. 
\begin{theorem}\label{the:poly}
	{\sc Opt-ISR} is solvable in linear time for chordal graphs.
\end{theorem}

This theorem can be obtained from the following lemma; 
we note that a maximum independent set $\indm$ of a chordal graph can be found in linear time~\cite{F75}, and the maximality of a given independent set can be checked in linear time. 
\begin{lemma}\label{lem:chordal_opt}
	Let $(G, \thr, \ind_\ini)$ be an instance of {\sc Opt-ISR} such that $G$ is a chordal graph, and let $\indm$ be any maximum independent set of $G$.
	Then, a solution $\ind_\opt$ to $(G, \thr, \ind_\ini)$ can be obtained as follows{\rm :} 
	\[
	\ind_\opt = \left\{
	\begin{array}{ll}
	\ind_\ini & ~~\mbox{if $\ind_\ini$ is a maximal independent set of $G$ and $|\ind_\ini| = \thr$}; \\
	\indm     & ~~\mbox{otherwise}.
	\end{array}
	\right.
	\]
\end{lemma}
\begin{proof}
	We first consider the case where $\ind_\ini$ is a maximal independent set of $G$ and $|\ind_\ini| = \thr$. 
	In this case, we cannot remove any vertex from $\ind_\ini$ because $|\ind_\ini| = \thr$.
	Furthermore, since $\ind_\ini$ is maximal, we cannot add any vertex in $V(G) \setminus \ind_\ini$ to $\ind_\ini$ while maintaining independence.
	Therefore, $G$ has no independent set $\ind^\prime$ $(\neq \ind_\ini)$ which is reachable from $\ind_\ini$, and hence $\ind_\opt = \ind_\ini$. 
	
	We then consider the other case, that is, $\ind_\ini$ is not a maximal independent set of $G$ or $|\ind_\ini| > \thr$. 
	Observe that it suffices to consider the case where $|\ind_\ini| > \thr$ holds; 
	if $|\ind_\ini|=\thr$, then $\ind_\ini$ is not maximal and hence we can obtain an independent set $\ind_\ini^{\prime \prime}$ of $G$ such that $|\ind_\ini^{\prime \prime}| = \thr + 1$ and  $\ind_\ini \sevsteptar{\thr} \ind_\ini^{\prime \prime}$ by adding some vertex in $V(G) \setminus \ind_\ini$.
	To prove $\ind_\opt = \indm$, we below show that $\ind_\ini \sevsteptar{\thr} \indm$ holds if $|\ind_\ini| > \thr$.
	
	Let $\ind^\prime_\ini \subseteq \ind_\ini$ be any independent set of size $\thr +1$.
	Then, $\ind_\ini \sevsteptar{\thr} \ind^\prime_\ini$ holds, because we can obtain $\ind^\prime_\ini$ from $\ind_\ini$ by removing vertices in $\ind_\ini \setminus \ind^\prime_\ini$ one by one.
	Similarly, let $\ind^\prime \subseteq \indm$ be any independent set of size $\thr + 1$; we know that $\ind^\prime \sevsteptar{\thr} \indm$ holds.
	Kami\'nski et al.~\cite{KMM12} proved that any two independent sets of the same size $\thr+1$ are reachable under the $\TAR{\thr}$ rule for even-hole-free graphs. 
	Since any chordal graph is even-hole free, we thus have $\ind^\prime_\ini \sevsteptar{\thr} \ind^\prime$. 
	Therefore, we have $\ind_\ini \sevsteptar{\thr} \ind^\prime_\ini \sevsteptar{\thr} \ind^\prime \sevsteptar{\thr} \indm$, and hence we can conclude that $\ind_\ini \sevsteptar{\thr} \indm$ holds as claimed.
\end{proof}

	We note that Lemma~\ref{lem:chordal_opt} indeed holds for even-hole-free graphs, which contain all chordal graphs. 
	However, the complexity status of the (ordinary) maximum independent set problem is unknown for even-hole-free graphs, and hence we do not know if we can obtain $\indm$ in polynomial time. 
	Indeed, Theorem~\ref{lem:maxind} implies that the complexity status of {\sc Opt-ISR} also remains open for even-hole-free graphs.

\section{Fixed Parameter Tractability} \label{sec:FPT}

In this section, we study the fixed parameter (in)tractability of {\sc Opt-ISR}. 
We take the solution size of {\sc Opt-ISR} as the parameter. 
More formally, for an instance $(G, \thr,\ind_\ini)$, the problem {\sc Opt-ISR} {\em parameterized by solution size} $\solsize$ asks whether $G$ has an independent set $\ind$ such that $|\ind| \ge \solsize$ and $\ind_\ini \sevsteptar{\thr} \ind$. 
We may assume that $\solsize > \thr$; otherwise it is a $\YES$-instance because $\ind_\ini$ itself is a solution. 
We sometimes denote by a $4$-tuple $(G, \thr, \ind_\ini, \solsize)$ an instance of {\sc Opt-ISR} parameterized by solution size $\solsize$. 

\subsection{Single parameter: solution size}\label{subsec:xp}
We first give an observation that can be obtained from {\sc independent set}. 
Because {\sc independent set} is W[1]-hard when parameterized by solution size $\solsize$~\cite{N06}, Lemma~\ref{lem:maxind} implies the following theorem.
\begin{theorem}\label{the:w1-hard}
	{\sc Opt-ISR} is $W[1]$-hard when parameterized by solution size $\solsize$. 
\end{theorem}

This theorem implies that {\sc Opt-ISR} admits no fixed-parameter algorithm with respect to solution size $\solsize$ under the assumption that ${\rm  FPT} \neq {\rm W[1]}$.
However, it admits an XP algorithm with respect to $\solsize$, as in the following theorem. 
\begin{theorem}\label{the:xp}
	For an instance $(G,\thr,\ind_\ini, \solsize)$, {\sc Opt-ISR} parameterized by solution size can be solved in time $O(\solsize^3 n^{2\solsize})$, where $n$ is the number of vertices in $G$.
\end{theorem}
\begin{proof}
	We construct an {\em auxiliary graph} $G_A$, defined as follows:
	Each node in $G_A$ corresponds to an independent set $\ind$ of $G$ such that $\thr \le |\ind| \le \solsize$, and there is an edge in $G_A$ between two nodes corresponding to independent sets $\ind$ and $\ind^\prime$ if and only if $|\symdiff{\ind}{\ind^\prime}|=1$ holds. 
	Notice that $G_A$ has a node corresponding to $\ind_\ini$, since $\thr \le |\ind_\ini| \le \solsize$.
	Then, by breadth-first search starting from the node corresponding to $\ind_\ini$, we can check if there is an independent set $\ind$ of $G$ such that $|\ind| = \solsize$ and $\ind_\ini \sevsteptar{\thr} \ind$. 
	
	We now estimate the running time of the algorithm. 
	Let $n$ and $m$ denote the numbers of vertices and edges in $G$, respectively. 
	The number of (candidates of) nodes in $G_A$ can be bounded by $\sum_{\thr \le j \le \solsize} {n \choose j} = O(\solsize n^{\solsize})$.
	For each enumerated vertex subset of $G$, we check if it forms an independent set of $G$; this can be done in time $O(n+m)$. 
	Therefore, the node set $V(G_A)$ can be constructed in time $O(\solsize n^{\solsize}(n+m))$.
	We then check each pair of nodes in $V(G_A)$; there are $O(|V(G_A)|^2) = O(\solsize^2 n^{2\solsize})$ pairs. 
	We join the pair by an edge in $G_A$ if their corresponding independent sets differ in only one vertex; we can check this condition in time $O(\solsize)$ for each pair of nodes.
	In this way, we can construct the auxiliary graph $G_A$ in time $O(\solsize^3 n^{2\solsize})$ in total. 
	Since breadth-first search can be executed in time $O(|V(G_A)|+|E(G_A)|) = O(\solsize^2 n^{2\solsize})$, our algorithm runs in time $O(\solsize^3 n^{2\solsize})$ in total.
\end{proof}

\subsection{Two parameters: solution size and degeneracy}\label{subsec:fpt}

As we have shown in Theorem~\ref{the:w1-hard}, {\sc Opt-ISR} admits no fixed-parameter algorithm when parameterized by the single parameter of solution size $\solsize$ under the assumption that ${\rm  FPT} \neq {\rm W[1]}$.
In addition, Theorem~\ref{the:pspace_complete} implies that the problem remains PSPACE-hard even if the degeneracy $\dg$ of an input graph is bounded by a constant, and hence {\sc Opt-ISR} does not admit even an XP algorithm with respect to the single parameter $\dg$ under the assumption that ${\rm P} \neq {\rm PSPACE}$. 
In this subsection, we take these two parameters, and develop a fixed-parameter algorithm as in the following theorem.
\begin{theorem}\label{the:fpt}
	{\sc Opt-ISR} admits a fixed-parameter algorithm when parameterized by $\solsize + \dg$, where $\solsize$ is the solution size and  $\dg$ is the degeneracy of an input graph.
\end{theorem}

Before proving the theorem, we note the following corollary which holds for planar graphs, and for bounded pathwidth graphs. 
Recall that {\sc Opt-ISR} is intractable (from the viewpoint of polynomial-time solvability) for these graphs, as shown in Theorems~\ref{the:np-hard-planar} and \ref{the:pspace_complete}.
\begin{corollary}\label{cor:FPT}
	{\sc Opt-ISR} parameterized by solution size $\solsize$ is fixed-parameter tractable for planar graphs, and for bounded treewidth graphs.
\end{corollary}
\begin{proof}
	Recall that the degeneracy of any planar graph is at most five.
	It is known that the degeneracy of a graph is at most the treewidth of the graph.
	Thus, the corollary follows from Theorem~\ref{the:fpt}.
\end{proof}

\subsubsection{Outline of algorithm}
As a proof of Theorem~\ref{the:fpt}, we give such an algorithm. 
We first explain our idea and the outline of the algorithm. 
Our idea is to extend a fixed-parameter algorithm for {\sc Reach-ISR} when parameterized by $\thr + \dg$~\cite{LMPRS15}.

If an input graph consists of only a fixed-parameter number of vertices, then we apply Theorem~\ref{the:xp} to the instance and obtain the answer in fixed-parameter time (Lemma~\ref{lem:brute-force}).  
We here use the fact (stated by Lokshtanov et al.~\cite[Proposition~3]{LMPRS15})  that a $\dg$-degenerate graph consists of a small number of vertices if it has a small number of low-degree vertices (Lemma~\ref{lem:dege}).

Therefore, it suffices to consider the case where an input graph has many low-degree vertices.
In this case, we will kernelize the instance: we will show that there always exists a low-degree vertex which can be removed from an input graph without changing the answer ($\YES$ or $\NO$) to the instance.  
Our kernelization has two stages.
In the first stage, we focus on ``twins'' (two vertices that have the same closed neighborhoods), and prove that one of them can be removed without changing the answer (Lemma~\ref{lem:reduce_neig}). 
The second stage will be executed only when the first stage cannot kernelize the instance to a sufficiently small size. 
The second stage is a bit involved, and makes use of the Sunflower Lemma by Erd\"os and Rado~\cite{ER60}.

\subsubsection{Graphs having a small number of low-degree vertices}

We now give our algorithm. 
Suppose that $(G,\thr,\ind_\ini,\solsize)$ is an instance of {\sc Opt-ISR} parameterized by solution size such that $G$ is a $\dg$-degenerate graph.
We assume that $|\ind_\ini| < \solsize$; otherwise $(G,\thr,\ind_\ini,\solsize)$ is a $\YES$-instance because $\ind_\ini$ itself is a solution. 

We first show the following property for $\dg$-degenerate graphs, which is a little bit stronger claim than that of Lokshtanov et al.~\cite[Proposition~3]{LMPRS15}; however, the proof is almost the same as that of~\cite{LMPRS15}.
\begin{lemma} \label{lem:dege}
	Suppose that a graph $G$ is $\dg$-degenerate, and let $\Dset \subseteq V(G)$ be the set of all vertices of degree at most $2\dg$ in $G$.
	Then, $|V(G)| \leq (2 \dg + 1)|\Dset|$.
\end{lemma}
\begin{proof}
	Suppose for a contradiction that $|V(G)| = (2 \dg + 1)|\Dset| + c$ holds for some integer $c \ge 1$. 
	Then, $|V(G) \setminus \Dset| = 2 \dg |\Dset|+c$, and hence we have
	\begin{eqnarray*}
		|E(G)| &=& \frac{1}{2} \sum_{v \in V(G)} |\oneig{G}{v}| \ge \frac{1}{2} \sum_{v \in V(G) \setminus \Dset} (2\dg + 1) \\
		&=& \frac{1}{2} (2\dg+1)  (2\dg |\Dset|+c)= \dg |V(G)| + \frac{1}{2}c > \dg |V(G)|. 
	\end{eqnarray*}
	This contradicts the fact that $|E(G)| \le \dg |V(G)|$ holds for any $\dg$-degenerate graph $G$~\cite{LW70}.
\end{proof}

Let $\Dset = \{v \in V(G) : |\oneig{G}{v}| \le 2\dg\}$, and let $\Dsetp = \Dset \setminus \ind_\ini$. 
We introduce a function $f(\solsize,\dg)$ which depends on only $\solsize$ and $\dg$; more specifically, let $f(\solsize,\dg) = \funcf$.	
We now consider the case where $G$ has only a fixed-parameter number of vertices of degree at most $2\dg$. 
\begin{lemma}\label{lem:brute-force}
	If $|\Dsetp| \leq f(\solsize,\dg)$, then {\sc Opt-ISR} can be solved in fixed-parameter time with respect to $\solsize$ and $\dg$.
\end{lemma}
\begin{proof}
	Since $\Dsetp = \Dset \setminus \ind_\ini$ and $|\ind_\ini| < \solsize$, we have $|\Dset| \le |\Dsetp|+|\ind_\ini| < f(\solsize, \dg) + \solsize$. 
	By Lemma~\ref{lem:dege} we thus have $|V(G)| \le (2 \dg + 1)|\Dset| < (2 \dg + 1)(f(\solsize, \dg) + \solsize)$. 
	Therefore, $|V(G)|$ depends on only $\solsize$ and $\dg$. 
	Then, this lemma follows from Theorem~\ref{the:xp}.
\end{proof}

\subsubsection{First stage of kernelization}

We now consider the remaining case, that is, $|\Dsetp| > f(\solsize,\dg)$ holds. 
The first stage of our kernelization focuses on ``twins'', two vertices having the same closed neighborhoods, and removes one of them without changing the answer.
\begin{lemma}\label{lem:reduce_neig}
	Suppose that there exist two vertices $b_i$ and $b_j$ in $\Dsetp$ such that $\cneig{G}{b_i} = \cneig{G}{b_j}$.
	Then, $(G,\thr,\ind_\ini,\solsize)$ is a $\YES$-instance if and only if $(G \setminus \set{b_i},\thr,\ind_\ini,\solsize)$ is.
\end{lemma}
\begin{proof}
	We note that $b_i \notin \ind_\ini$ and $b_j \notin \ind_\ini$, because $\Dsetp = \Dset \setminus \ind_\ini$.
	Then, the if direction is trivial, and hence we prove the only-if direction.
	Suppose that $(G,\thr,\ind_\ini,\solsize)$ is a $\YES$-instance, and hence $G$ has an independent set $\ind_\opt$ such that $|\ind_\opt| \ge \solsize$ and $\ind_\ini \sevsteptar{\thr} \ind_\opt$.
	Then, there exists a reconfiguration sequence $\mathcal{\ind} = \seq{\ind_\ini,\ind_1,\ldots,\ind_\ell=\ind_\opt}$.
	Since $\cneig{G}{b_i} = \cneig{G}{b_j}$, we know that $b_i$ and $b_j$ are adjacent in $G$ and hence no independent set of $G$ contains $b_i$ and $b_j$ at the same time. 
	We now consider a new sequence $\mathcal{\ind}^\prime = \seq{\ind_\ini^\prime,\ind_1^\prime,\ldots,\ind_\ell^\prime}$ defined as follows:
	for each $x \in \{\ini, 1, \ldots, \ell \}$, let
	\[
	\ind_x^\prime = \left\{ 
	\begin{array}{ll}
	\ind_x & ~~~\mbox{if $b_i \notin \ind_x$};\\
	(\ind_x \setminus \set{b_i} ) \cup \set{b_j} & ~~~\mbox{otherwise}.
	\end{array}
	\right.
	\]
	Since each $\ind_x$,  $x \in \{\ini, 1, \ldots, \ell \}$, is an independent set of $G$ and $\cneig{G}{b_i} = \cneig{G}{b_j}$, each $\ind_x^\prime$ forms an independent set of $G$.
	In addition, since $|\symdiff{\ind_{x-1}}{\ind_{x}}| = 1$ for all $x \in \{1,2,\ldots, \ell\}$, we have $|\symdiff{\ind_{x-1}^\prime}{\ind_{x}^\prime}| = 1$.
	Therefore, $\mathcal{\ind}^\prime$ is a reconfiguration sequence such that no independent set in $\mathcal{\ind}^\prime$ contains $b_i$.
	Since $|\ind_\ell^\prime| = |\ind_\ell| = |\ind_\opt| \ge \solsize$, we can conclude that $(G \setminus \set{b_i},\thr,\ind_\ini,\solsize)$ is a $\YES$-instance. 
\end{proof}

We repeatedly apply Lemma~\ref{lem:reduce_neig} to a given graph, and redefine $G$ as the resulting graph; we also redefine $\Dset$ and $\Dsetp$ according to the resulting graph $G$. 
Then, any two vertices $b_i$ and $b_j$ in $\Dsetp$ satisfy $\cneig{G}{b_i} \neq \cneig{G}{b_j}$.
If $|\Dsetp| \le f(\solsize,\dg)$, then we have completed our kernelization; recall Lemma~\ref{lem:dege}. 
Otherwise, we will execute the second stage of our kernelization described below.

\subsubsection{Second stage of kernelization}		
In the second stage of the kernelization, we use the classical result of Erd\"os and Rado~\cite{ER60}, known as the {\em Sunflower Lemma}.
We first define some terms used in the lemma.
Let $\Petal_1, \Petal_2, \ldots, \Petal_\npet$ be $\npet$ non-empty sets over a universe $U$, and let $\Core \subseteq U$ which may be an empty set. 
Then, the family $\{ \Petal_1, \Petal_2, \ldots, \Petal_\npet \}$ is called a {\em sunflower} with a {\em core} $\Core$ if $\Petal_i \setminus \Core \neq \emptyset$ holds for each $i \in \set{1,2,\ldots,\npet}$, and $\Petal_i \cap \Petal_j = \Core$ holds for each $i,j \in \set{1,2,\ldots,\npet}$ satisfying $i \neq j$.
The set $\Petal_i \setminus \Core$ is called a {\em petal} of the sunflower.
Note that a family of pairwise disjoint sets always forms a sunflower (with an empty core).
Then, the following lemma holds.
\begin{lemma}[Erd\"os and Rado~\cite{ER60}]\label{lem:sunflower}
	Let $\mathcal{A}$ be a family of sets {\rm (}without duplicates{\rm )} over a universe $U$ such that each set in $\mathcal{A}$ is of size at most $t$.
	If $|\mathcal{A}| > t!(\npet-1)^t$, then there exists a family $\Sunf \subseteq \mathcal{A}$ which forms a sunflower having $\npet$ petals.
	Furthermore, $\Sunf$ can be computed in time polynomial in $|\mathcal{A}|$, $|U|$, and $\npet$.
\end{lemma}

We now explain the second stage of our kernelization, and make use of Lemma~\ref{lem:sunflower}. 
Let $b_1,b_2,\ldots,b_{|\Dsetp|}$ denote the vertices in $\Dsetp$, and let $\mathcal{A} = \set{\cneig{G}{b_1},\cneig{G}{b_2},\ldots, \cneig{G}{b_{|\Dsetp|}}}$ be the set of closed neighborhoods of all vertices in $\Dsetp$.
In the second stage, recall that $\cneig{G}{b_i} \neq \cneig{G}{b_j}$ holds for any two vertices $b_i$ and $b_j$ in $\Dsetp$, and hence no two sets in $\mathcal{A}$ are identical. 
We set $U = \bigcup_{b_i \in \Dsetp}\cneig{G}{b_i}$.
Since each $b_i \in \Dsetp$ is of degree at most $2\dg$ in $G$, each $\cneig{G}{b_i} \in \mathcal{A}$ is of size at most $2\dg + 1$.
Notice that $|\mathcal{A}| = |\Dsetp| > f(\solsize,\dg) = \funcf$. 
Therefore, we can apply Lemma~\ref{lem:sunflower} to the family $\mathcal{A}$ by setting $t = 2 \dg+1$ and $\npet = 2\solsize + \dg + 1$, and obtain a sunflower $\Sunf \subseteq \mathcal{A}$ with a core $\Core$ and $\npet$ petals in time polynomial in $|\mathcal{A}|$, $|U|$, and $\npet = 2\solsize + \dg + 1$.
Let $\Sset = \set{b_1^\prime, b_2^\prime, \ldots, b_{\npet}^\prime} \subseteq \Dsetp$ be the set of $\npet$ vertices whose closed neighborhoods correspond to the sunflower $\Sunf$, that is, $\Sunf = \set{\cneig{G}{b_1^\prime},\cneig{G}{b_2^\prime},\ldots,\cneig{G}{b_{\npet}^\prime}} \subseteq \mathcal{A}$.
The following lemma says that at least $2\solsize$ vertices in $\Sset$ are contained in the $\npet$ petals of the sunflower $\Sunf$ (i.e., not in the core $\Core$), and they forms an independent set of $G$.
\begin{lemma}\label{lem:petal}
	$\Sset \setminus \Core$ is an independent set of $G$, and $\Sset \cap \Core$ forms a clique in $G$.
	Furthermore, $|\Sset \setminus \Core| \ge 2\solsize$. 
\end{lemma}
\begin{proof}
	To prove the lemma, we first claim that two distinct vertices $b_i^\prime, b_j^\prime \in \Sset$ are adjacent in $G$ if and only if $b_i^\prime, b_j^\prime \in \Core$. 
	Recall that their closed neighborhoods $\cneig{G}{b_i^\prime}$ and $\cneig{G}{b_j^\prime}$ belong to the sunflower $\Sunf$ with the core $\Core$, and hence $\cneig{G}{b_i^\prime} \cap \cneig{G}{b_j^\prime} = \Core$ holds. 
	In the if direction proof of the claim, since $b_i^\prime \in \Core = \cneig{G}{b_i^\prime} \cap \cneig{G}{b_j^\prime}$, we have $b_i^\prime \in \cneig{G}{b_j^\prime}$ and hence $b_i^\prime$ and $b_j^\prime$ are adjacent in $G$. 
	On the other hand, in the only-if direction proof of the claim, since $b_i^\prime$ and $b_j^\prime$ are adjacent in $G$, we have $b_i^\prime, b_j^\prime \in \cneig{G}{b_i^\prime} \cap \cneig{G}{b_j^\prime} = \Core$. 
	In this way, the claim holds.
	
	This claim indeed implies that $\Sset \setminus \Core$ is an independent set of $G$, and $\Sset \cap \Core$ forms a clique in $G$.  
	Therefore, to complete the proof of this lemma, it suffices to prove that $|\Sset \setminus \Core| \ge 2\solsize$ holds. 
	To see this, notice that $|\Sset \cap \Core| \le \dg + 1$ holds, because otherwise $\Sset \cap \Core$ forms a clique in $G$ of size at least $\dg + 2$; this contradicts the assumption that $G$ is a $\dg$-degenerate graph. 
	Then, $|\Sset \setminus \Core| = |\Sset| - |\Sset \cap \Core| \ge \npet - (\dg+1) = 2 \solsize$. 
\end{proof}

We are now ready to give the following lemma, as the second stage of the kernelization. 
\begin{lemma}\label{lem:reduce}
	Let $\rmv$ be any vertex in $\Sset$. 
	Then, $(G, \thr, \ind_\ini, \solsize)$ is a $\YES$-instance if and only if $(G \setminus \set{\rmv}, \thr, \ind_\ini, \solsize)$ is.
\end{lemma}
\begin{proof}
	Note that $\rmv \notin \ind_\ini$ since $\Sset \subseteq \Dsetp \subseteq V(G) \setminus \ind_\ini$.
	Then, the if direction clearly holds, and hence we prove the only-if direction.
	Suppose that $(G,\thr,\ind_\ini,\solsize)$ is a $\YES$-instance, and hence $G$ has an independent set $\ind_\opt$ such that $|\ind_\opt| \ge \solsize$ and $\ind_\ini \sevsteptar{\thr} \ind_\opt$.
	Then, there exists a reconfiguration sequence $\mathcal{\ind} = \seq{\ind_\ini,\ind_1,\ldots,\ind_\ell=\ind_\opt}$.
	If no independent set in $\mathcal{\ind}$ contains $\rmv$, then the only-if direction holds.
	Therefore, we consider the case where at least one independent set in $\mathcal{\ind}$ contains $\rmv$.
	Let $\ind_{r+1}$ be the first independent set in $\mathcal{\ind}$ which contains $\rmv$, that is, $\rmv \notin \ind_i$ for all $i \in \set{0,1,\ldots,r}$.
	We assume without loss of generality that $|\ind_i| < \solsize$ holds for all $i \in \set{0,1,\ldots,\ell-1}$.

	\begin{figure}[t]
	\centering
	\includegraphics[width=0.5\linewidth]{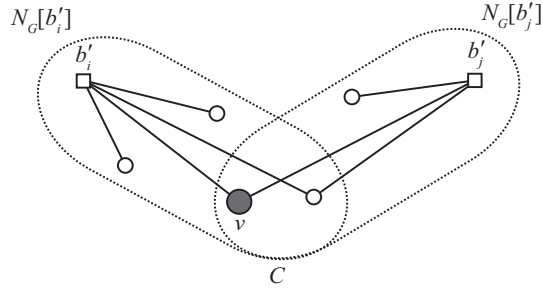}
	\caption{Illustration for Lemma~\ref{lem:reduce}, where two vertices $b_i^\prime, b_j^\prime \in \Sset \setminus \Core$ are depicted by squares. By the definition of a sunflower, all vertices $v$ adjacent to both $b_i^\prime$ and $b_j^\prime$ must be contained in $\Core = \cneig{G}{b_i^\prime} \cap \cneig{G}{b_j^\prime}$.}
	\label{fig:sunflower}
	\end{figure}
	
	Let $\Ssetp = \Sset \setminus (\set{\rmv} \cup \Core \cup \cneig{G}{\ind_r})$, where $\cneig{G}{\ind_r} = \bigcup_{v \in \ind_r} \cneig{G}{v}$. 
	We now claim that $\ind_\opt^\prime = \ind_r \cup \Ssetp$ is an independent set of $G$ such that $|\ind_\opt^\prime| \ge \solsize$ and $\ind_\ini \sevsteptar{\thr} \ind_\opt^\prime$ on $G \setminus \set{\rmv}$; then $(G \setminus \set{\rmv}, \thr, \ind_\ini, \solsize)$ is a $\YES$-instance.
	Since $\Ssetp \subseteq \Sset \setminus \Core$, Lemma~\ref{lem:petal} says that $\Ssetp$ is an independent set of $G$. 
	Furthermore, since $\Ssetp$ does not contain any vertex in $\cneig{G}{\ind_r}$, $\ind_\opt^\prime = \ind_r \cup \Ssetp$ is an independent set of $G$. 
	Recall that $\ind_\ini \sevsteptar{\thr} \ind_r$ holds on $G \setminus \set{\rmv}$, and hence we know $|\ind_r| \ge \thr$. 
	Then, $\ind_\ini \sevsteptar{\thr} \ind_r \sevsteptar{\thr} \ind_\opt^\prime$ holds on $G \setminus \set{\rmv}$ by adding the vertices in $\Ssetp$ to $\ind_r$ one by one. 
	We finally prove that $|\ind_\opt^\prime| \ge \solsize$ by showing that $|\Ssetp| \ge \solsize$ holds.
	Since $|\Sset \setminus \Core| \ge 2\solsize$ (by Lemma~\ref{lem:petal}) and $|\ind_r| < \solsize$, it suffices to prove that $|\cneig{G}{v} \cap (\Sset \setminus \Core)| \le 1$ holds for each vertex $v \in \ind_r$. 
	Since $\ind_{r+1}$ is obtained by adding $\rmv$ to $\ind_{r}$, we know $\ind_r \cap \cneig{G}{\rmv} = \emptyset$. 
	Since $\Core \subset \cneig{G}{\rmv}$, we thus have $\ind_r \cap \Core = \emptyset$. 
	Therefore, each vertex $v \in \ind_r$ is adjacent to at most one vertex in $\Sset \setminus \Core$, because otherwise $v$ must be contained in $\Core$. 
(See also \figurename~\ref{fig:sunflower}.)
\end{proof}

We can repeatedly apply Lemma~\ref{lem:reduce} to $G$ until the resulting graph has at most $ f(\solsize,\dg)$ vertices of degree at most $2\dg$. 
Then, by Lemma~\ref{lem:dege} we have completed our kernelization.

This completes the proof of Theorem~\ref{the:fpt}.

\section{Conclusions}
In this paper, we have introduced a new framework for reconfiguration problems, and applied it to {\sc independent set reconfiguration}.
As shown in \figurename~\ref{fig:results} and \tablename~\ref{tab:result_parameter}, we have studied the problem from the viewpoints of polynomial-time solvability and the fixed-parameter (in)tractability, and shown several interesting contrasts among graph classes and parameters. 
In particular, we gave a complete analysis of the fixed-parameter (in)tractability with respect to the three parameters.

\subsection*{Acknowledgments}
We thank Tatsuhiko Hatanaka for his insightful suggestions on this new framework.
We are grateful to Tesshu Hanaka, Benjamin Moore, Vijay Subramanya, and Krishna Vaidyanathan for valuable discussions with them. 
Research by Japanese authors is partially supported by JST CREST Grant Number JPMJCR1402, and JSPS KAKENHI Grant Numbers JP16K00004 and JP17K12636, Japan. 
Research by Naomi Nishimura is supported by the Natural Science and Engineering Research Council of Canada.

\bibliographystyle{abbrv}
\bibliography{references}

\end{document}

%% file: newcom.tex
\newenvironment{listing}[1]{%
	\begin{list}{*}{%
			\settowidth{\labelwidth}{#1}%
			\setlength{\leftmargin}{\labelwidth}%
			\advance \leftmargin by 12pt
			\setlength{\itemsep}{0pt}%
			\setlength{\parsep}{0pt}%
			\setlength{\topsep}{0pt}%
			\setlength{\parskip}{0pt}%
		}%
	}{%
\end{list}}

\newcounter{one}
\newcounter{two}
\newcounter{three}
\newcounter{four}
\newcounter{five}
\newcommand{\YES}{\mathsf{yes}}
\newcommand{\NO}{\mathsf{no}}

\newcommand{\set}[1]{\{ #1 \}}
\newcommand{\seq}[1]{\langle #1 \rangle}

\newcommand{\symdiff}[2]{#1 \triangle #2}

\newcommand{\bag}{X}
\newcommand{\oneig}[2]{ N_{ #1 }( #2 ) }
\newcommand{\cneig}[2]{ N_{ #1 }[ #2 ] }
\newcommand{\dg}{d}

\newcommand{\sevstep}{\leftrightsquigarrow} 
\newcommand{\sevsteptar}[1]{\overset{#1}{\sevstep}} 

\newcommand{\TARrule}{\mathsf{TAR}}
\newcommand{\TAR}[1]{\TARrule(#1)}
\newcommand{\TS}{\mathsf{TS}}
\newcommand{\TJ}{\mathsf{TJ}}

\newcommand{\ini}{0} 
\newcommand{\tar}{r} 
\newcommand{\opt}{{\sf sol}} 
\newcommand{\thr}{l}
\newcommand{\solsize}{s}

\newcommand{\ind}{I} 
\newcommand{\indm}{\ind_{\sf max}} 

\newcommand{\funcf}{(2\dg + 1)!((2\solsize + \dg + 1) - 1)^{2\dg + 1}}
\newcommand{\rmv}{b_q^\prime}

\newcommand{\Dset}{D}
\newcommand{\Dsetp}{\Dset^\prime}
\newcommand{\Sset}{S}
\newcommand{\Ssetp}{\Sset^\prime}

\newcommand{\Sunf}{\mathcal{S}}
\newcommand{\Petal}{P}
\newcommand{\Core}{C}
\newcommand{\npet}{p}